\documentclass[journal,12pt,onecolumn,draftclsnofoot,]{IEEEtran}
\usepackage{cite}
\usepackage{amsmath,amssymb,amsfonts,amsthm}
\usepackage{mathrsfs}
\usepackage{array}
\usepackage{tabularx}
\usepackage{multirow}
\usepackage{adjustbox}
\usepackage{lipsum}
\usepackage{booktabs}
\usepackage{color,colortbl}
\usepackage{mathtools}
\usepackage{subfigure}
\usepackage{mathtools}
\usepackage{graphicx}
\usepackage{bbm}
\usepackage[utf8]{inputenc}
\usepackage{soul}
\usepackage{algorithmic,float}
\usepackage[linesnumbered,ruled,vlined]{algorithm2e}
\SetKwInput{KwInput}{Input}                
\SetKwInput{KwOutput}{Output}
\SetKwInput{KwInit}{Initialization}
\newtheorem{theorem}{Theorem}
\newtheorem{remark}{Remark}
\newtheorem{corollary}{Corollary}[theorem]
\usepackage{textcomp}
\usepackage[normalem]{ulem}
\usepackage{multicol}
\usepackage[ruled,vlined]{algorithm2e}%linesnumbered,
\SetKwInput{KwInput}{Input}           % Set the Input
\SetKwInput{KwOutput}{Output}
\SetKwInput{KwInit}{Initialization}
\allowdisplaybreaks

\begin{document}

%\title{Outage Probability Analysis of RIS-aided MISO Systems}
\title{Characterization of Capacity and Outage of RIS-aided Downlink Systems under Rician Fading}
\author{Kali Krishna Kota, Praful D. Mankar, Harpreet S. Dhillon 
        % <-this % stops a space
\thanks{K.K. Kota and P. D. Mankar are with Signal Processing and Communication Research Center, IIIT Hyderabad, India. (Email: kali.kota@research.iiit.ac.in,   praful.mankar@iiit.ac.in). H. S. Dhillon is with Wireless@VT, Department of ECE, Virginia Tech, Blacksburg, VA (Email: hdhillon@vt.edu). The work of H. S. Dhillon was supported by the U.S. National Science Foundation under Grants ECCS-2030215 and CNS-2225511.}% <-this % stops a space
\thanks{}}

\maketitle

\begin{abstract} 
This letter presents optimal beamforming and outage analysis for a Reconfigurable Intelligent Surface (RIS)-aided multiple input single output downlink system under Rician fading on both the direct and the RIS-assisted indirect links. We focus on maximizing the capacity for two transmitter architectures: fully digital (FD) and fully analog (FA). This capacity maximization problem with optimally configured RIS is shown to be $L_1$ norm-maximization with respect to the transmit beamformer. To obtain the optimal FD beamformer, we propose a complex $L_1$-PCA-based algorithm whose complexity is significantly lower than the existing semi-definite relaxation-based solutions. We also propose a low-complexity optimal beamforming algorithm to obtain the FA beamformer solution.
Further, we derive analytical upper bounds on the SNR achievable by the proposed algorithms and utilize them to characterize the lower bounds on outage probabilities. 
The derived bounds are numerically shown to closely match the achievable performance for a low-rank channel matrix and are shown to be exact for a unit-rank channel matrix. %We further prove the bounds to be exact for a unit-rank channel matrix, i.e., when the LoS components of DL and IL are strong and aligned.

% {\color{blue} This letter presents optimal beamforming and outage analysis of a RIS-aided multiple input single output system under Rician faded direct link and the RIS-assisted indirect link. We  show that the capacity maximization with optimally configured RIS becomes a $L_1$ norm-maximization in terms of transmit beamformer. We propose a complex $L_1$ PCA-based algorithm to obtain the fully digital beamformer and another algorithm to obtain the fully analog beamformer with low complexity. Further, we derive accurate bounds on capacity and outage and show that they are close/exact under a low rank/unit rank propagation environment.}
%%100 word abstract
% {\color{blue} This letter presents optimal beamforming and outage analysis for a RIS-aided multiple input single output downlink system under Rician faded direct and RIS-assisted indirect links.  We show that the capacity maximization problem with optimally configured RIS becomes  $L_1$ norm-maximization in terms of transmit beamformers. We devised novel low complexity algorithms to obtain optimal beamformers for fully-digital and fully-analog transmitter archieture. Further, we derive bounds on capacity and outage and show that they are close/exact for a low rank/unit rank channel matrix.}
\end{abstract}
\section{Introduction}
\IEEEPARstart{R}{econfigurable} intelligent surfaces (RISs) have attracted significant attention in recent years because of their ability to partially control the propagation environment and hence improve the performance of communications systems
% continue to be investigated because of their ability to improve propagation environments, which in turn can be utilized to improve spatial multiplexing and enhance coverage 
\cite{Wireless_through_RIS,RIS_principles&opp,CommModels_for_RIS}. 
%{\color{red}Despite the availability of extensive literature on the optimal design of RIS-aided communication systems, there is still a need to understand the capacity and outage performance of such systems in various fading environments and with transceiver architectures. } \\
%{\color{blue} Extensive literature on capacity and outage performance of such RIS-aided communication systems in various fading environments is available.}
% The performance-achieving beamforming and precoding algorithms for a variety of  RIS-aided communication system configurations have been extensively studied in the literature. For example, please refer to these excellent articles \cite{} and references therein. 
An extensive literature survey on the design of beamformers/precoders for RIS-aided multiple input single output/multiple input multiple output (MISO/MIMO) communication systems maximizing capacity is available in \cite{CAP1,CAP2,CAP3}.
% {\color{red}However, these works lack in terms of performance analysis of the proposed solutions, particularly for advanced settings, including multi-user and/or multiple antennas at the transmitter and receiver, because of analytical intractability. }
%%%
% However, these works lack in terms of performance analysis of the proposed solutions, particularly for multi antenna systems {\color{blue}and  non-zero mean channel} due to analytical intractability. Nonetheless, there are a few works on the characterization of capacity and outage performance under simplistic settings, which we discuss below.
%%%
A key shortcoming of the prior art in this direction is the lack of analytical performance characterization of the proposed solutions, especially for multi antenna systems, which is the main inspiration behind this paper. There are just a handful of works focusing on the characterization of capacity and outage performance, {\em albeit} under simplistic settings, which we discuss below.

%%%%%%%%%%%%%%%%%%%%%%%%%%%%%%%%%%%%%%%%%%%%%%%%%%%%%%%%%%%%%%%%%%%%%%%%%%%%%%%%%%%%%%%%%%%%%%%%%%%%%%%%%
%SISO
For RIS-aided single input single output (SISO) systems, the authors of \cite{RISvsRelay,P2,SISO_Capacity_Asymptotic,VanChien,PengXu,Tao,Chernoff_Saddlepoint} analyze the outage probability with/without the presence of a direct link (DL) between the transmitter and receiver in various fading environments. In particular, asymptotic outage probability is derived for Rayleigh fading in \cite{RISvsRelay} and \cite{P2} and for Rician fading in \cite{SISO_Capacity_Asymptotic} in the presence of only the RIS-assisted indirect link (IL), which is further used to analyze the diversity order and asymptotic symbol error rate. The authors of \cite{VanChien} derived a closed-form expression for the outage with Rayleigh fading along both DL and IL. Next, \cite{PengXu} extended the result presented in \cite{VanChien} for Rician fading along IL. 
For similar settings, \cite{Tao,Chernoff_Saddlepoint} derived upper bounds on the outage, and \cite{Chernoff_Saddlepoint} obtained the asymptotically exact outage in closed form.
%On the other hand, \cite{Chernoff_Saddlepoint} and \cite{Tao} derived upper bounds on outage for similar settings. Besides, \cite{Chernoff_Saddlepoint} also presented a closed-form solution of the asymptotic outage probability.  
%%%%%%%%%%%%%%%%%%%%%%%%%%%%%%%%%%%%%%%%%%%%%%%%%%%%%%%%%%%%%%%%%%%%%%%%%%%%%%%%%%%%%%%%%%%%%%%%%%%%%%%%%

%MISO
In addition, few works focused on outage analysis for RIS-aided MISO systems. 
For example,  \cite{Neel} and \cite{Guo} derived outage probabilities under various channel models for a maximum ratio transmission (MRT)-based transmit beamformer and optimally configured RIS phase shifts.
The authors of \cite{Neel} considered the line-of-sight (LoS)-Rayleigh channel model along IL, whereas \cite{Guo} considered LoS-Rician channel along IL and Rician channel along DL. 
%The authors of \cite{Neel} and \cite{Guo} employed maximum ratio transmission (MRT)-based transmit beamformer along with optimally configured RIS and derived outage probability for various channel models. 
% \cite{Neel,SlimAlouini_Letter,Guo,IIIT1,2024} present . 
% More specifically, \cite{Guo,Neel} adopt the maximum ratio transmission (MRT) scheme for the transmit beamformer in the presence of LoS channel between BS-RIS.     
%%%%%%%%%%%%%%%%%%%%%%%%%%%%%%%%%%%%%%%%%%%%%%%%%%%%%%%%%%%%%%%%%%%%%%%%%%%%%%%%%%%%%%%%%%%%%%%%%%%%%%%%%
%MIMO
%%%%%%%%%%%%%%%%%%%%%%%%%%%%%%%%%%%%%%%%%%%%%%%%%%%%%%%%%%%%%%%%%%%%%%%%%%%%%%%%%%%%%%%%%%%%%%%%%%%%%%%%%
%This paper is not for rank improvement as such. This paper proposes discrete RIS phase shift design for rank deficient channels
On the other hand, \cite{SlimAlouini_Letter} analyzed the capacity of the MRT-based transmit beamformer and optimally configured discrete RIS phase shifts for a millimeter wave channel.
% proposed an algorithm to obtain the optimal phase shift solution by exploiting low-rank structure and maximum ratio transmission (MRT) as the transmit beamforming solution that maximized the capacity of a millimeter wave channel.
% Statistical CSI 
%%%%%%%%%%%%%%%%%%%%%%%%%%%%%%%%%%%%%%%%%%%%%%%%%%%%%%%%%%%%%%%%%%%%%%%%%%%%%%%%%%%%%%%%%%%%%%%%%%%%%%%%% 
Next, \cite{IIIT1} presented statistically optimal beamforming for maximizing the ergodic capacity upper bound under Rician fading along both DL and IL and also characterized its outage performance. Finally, the ergodic capacity performance of a RIS-aided MIMO system with no DL is investigated in \cite{2024}. Particularly, they derived a closed-form expression of the channel gain distribution under Rayleigh-Rician fading along IL under full-rank and low-rank channel scenarios and utilized it to obtain ergodic capacity.
To summarize, the above-mentioned works focus on the outage/capacity characterization of RIS-aided (mostly SISO and MISO to some extent) systems under different fading scenarios. However, the analyses of MISO systems conveniently consider channel models that are tractable.
In addition, there is no consideration of the transmitter architecture type and its impact on capacity. Inspired by this, we attempt to bridge these gaps by considering Rician fading along both DL and IL with different transmitter architectures, namely fully digital (FD) and fully analog (FA). Our key contributions are listed below. 
%%Such an analysis is important while making hardware design choices that balance hardware complexity and performance tradeoffs. \\
% {\color{blue}
%     $\bullet$ 
%     We first show that the capacity maximization with optimally configured RIS becomes $L_1$ norm maximization problem with respect to the transmit beamforming vector.\\
%     $\bullet$ For FD architecture, we proposed an $L_1$-PCA for complex data problem based beamforming algorithm which has complexity of $\mathcal{O}(MN)$ that is significantly lower compared to the existing solutions in the literature.\\
%     $\bullet$ We derive an upper bound on the capacity for both FA and FD architectures based optimal beamformings. Next, this bound is shown to be achievable  when a strong LoS component is present along the IL.\\
%     $\bullet$ We also derive the  moment generating function (MGF) of the envelope of upper bounded SNR and numerically invert it to obtain a lower bound on outage probability. We further show that this lower bound is exact when the channel matrix is unit-rank, which indeed is the case when RIS is utilized to provide an virtual LoS for coverage enhancement. %This result is particularly important as RIS is envisioned to provide a virtual LoS channel with rank 1.   
% }
\begin{enumerate}
    \item We first show that the capacity maximization problem with optimally configured RIS becomes an $L_1$ norm maximization problem with respect to the transmit beamforming vector.
    \item For FD architecture, we propose an optimal beamforming algorithm based on complex $L_1$-PCA, which has significantly lower complexity compared to the existing semi-definite relaxation (SDR)-based solutions.
    \item We derive an upper bound on the capacity for both the architectures that is shown to be achievable in the presence of strong LoS along the IL and  absence of DL.
    \item In the absence of DL, we also derive the moment generating function (MGF) of the envelope of the upper bounded SNR and numerically invert it to obtain the outage probability lower bound. We further show that this lower bound is exact when the channel matrix is unit-rank, which is the case when RIS is utilized to provide a virtual LoS for coverage enhancement. 
\end{enumerate}
\section{System Model}\label{system_model} 
% This letter considers a RIS-aided MISO downlink system wherein the transmitter equipped with $M$ antennas sends information to a single antenna user via RIS comprising $N$ reflecting elements. We consider two types of transmitter configurations: 1) FD architecture, wherein each antenna is connected to the baseband processing unit (BBU) via a dedicated radio frequency (RF) chain, and 2) FA architecture, wherein all the antennas are connected to the BBU via a single RF chain. Furthermore, we consider the transmitter to efficiently utilize the {\em DL} and the RIS-assisted {\em IL} to improve the system capacity. We also assume both the links to have strong LoS and multi-path fading components. To model such scenarios, we consider Rician fading model along BS-RIS, RIS-user, and BS-user links given by
This letter considers a RIS-aided MISO downlink system with a single antenna user, wherein the transmitter is equipped with $M$ antennas, and RIS comprises of $N$ reflecting elements. We consider two types of transmitter configurations: 1) FD architecture, wherein each antenna is connected to the baseband processing unit (BBU) via a dedicated radio frequency (RF) chain, and 2) FA architecture, wherein all the antennas are connected to the BBU via a single RF chain. We also consider the transmitter to efficiently utilize the Rician faded {\em DL} and {\em IL}. Thus, BS-RIS, RIS-user, and BS-user links are modeled as 
$$ \mathbf{H} = \kappa_{l1}\mathbf{\bar{H}} + \kappa_{n1}\mathbf{\tilde{H}},~
    \mathbf{h}= \kappa_{l2}\mathbf{\bar{h}} + \kappa_{n2}\mathbf{\tilde{h}}\text{~and~}\mathbf{g}= \kappa_{lo}\mathbf{\bar{g}} + \kappa_{n0} \mathbf{\tilde{g}},$$
respectively, where $\kappa_{li} = \sqrt{\frac{K_i}{1+K_i}}, \kappa_{ni} = \sqrt{\frac{1}{1+K_i}}$. 
% {\color{red}Moreover, $K_o$, $K_1$, and $K_2$ denote the fading factors associated with the BS-user, BS-RIS, and RIS-user links, respectively, and  $\mathbf{\Bar{g}}$/$\mathbf{\tilde{g}}$,   $\mathbf{\Bar{h}}$/$\mathbf{\tilde{h}}$,   $\mathbf{\Bar{H}}$/$\mathbf{\tilde{H}}$ denote the LoS/multipath components associated with the BS-user, BS-RIS, and RIS-user links, respectively.}
Moreover, the tuple $\{K_o,\mathbf{\Bar{g}},\mathbf{\tilde{g}}\}$,   $\{K_1,\mathbf{\Bar{h}}, \mathbf{\tilde{h}}\}$,   $\{K_2,\mathbf{\Bar{H}}, \mathbf{\tilde{H}}\}$ consist of the fading factor, LoS component, and multipath component associated with the BS-user, BS-RIS, and RIS-user links, respectively.
Further, we model the LoS components as
\begin{align*}
\mathbf{\bar{g}}&=\mathbf{a}_M(\theta_{\rm{bd}}^d),~ \mathbf{\bar{h}}=\mathbf{a}_N(\theta_{\rm{rd}}),
~\text{and}~\mathbf{\bar{H}}=\mathbf{a}_N(\theta_{\rm ra})\mathbf{a}_M(\theta_{\rm bd}^i)^T,
\end{align*} where $\mathbf{a}_L(\psi)=[1 e^{-j\pi\frac{d}{\lambda}\cos(\psi)} \dots e^{-j\pi\frac{d}{\lambda}(L-1)\cos(\psi)}]^T$ is the linear array response, and $\theta_{\rm bd}^{d}$, $\theta_{\rm bd}^i$, $\theta_{\rm rd}$ are the angles of departure of the LoS components along the BS-RIS, RIS-BS, and RIS-user links, respectively and $\theta_{\rm ra}$ is the arrival angle at the RIS from BS.    
The multipath components are modeled as $\Tilde{\mathbf{g}} \sim \mathcal{CN}(0,\mathbf{I}_M)$, $\Tilde{\mathbf{h}} \sim \mathcal{CN}(0,\mathbf{I}_N)$, and $\Tilde{\mathbf{H}}_{:,i} \sim \mathcal{CN}(0,\mathbf{I}_N)$.

The received signal at the user is 
\begin{align}
    y = l(d_1,d_2) \mathbf{h}^T\mathbf{\Phi Hf}x + l(d_0) \mathbf{g}^T \mathbf{f}x +n,\label{RxSig} %\mathbf{g}^T\mathbf{f}x
\end{align}
where $x\in\mathbb{C}$ is the transmit symbol with $\mathbb{E}[xx^*] = P_s$, $P_s$ is the total power available at the transmitter, $\mathbf{f}\in\mathbb{C}^M$ is the transmit beamforming vector, $\mathbf{\Phi}={\rm diag}(\boldsymbol{\psi})$ is the RIS phase shift matrix, $l(d_1,d_2) = (d_1 d_2)^{-\alpha/2}, l(d_0) = d_0^{-\alpha/2}$ are the path loss models along IL and DL, $d_0, d_1, d_2$ are the distances between BS-user, BS-RIS, and RIS-user, respectively, $\alpha$ is the path loss exponent, and $n \sim \mathcal{CN}(0,\sigma_n^2)$ is the complex Gaussian noise. Further, the transmit beamforming vector belongs to 1) $ \mathcal{B}  = \{ \mathbf{f} \in \mathbb{C}^M :\|\mathbf{f}\|_2 = 1\}$ under FD architecture to limit power consumption beyond the total available power and 2) $\mathcal{L} = \{\mathbf{f} \in \mathbb{C}^M :|\mathbf{f}_m| = 1/\sqrt{M}\}$ under the FA architecture, implying the unit modulus constraint. Next, $\boldsymbol{\psi} \in \mathbb{C}^N$ is the RIS phase shift vector with a unit amplitude constraint, i.e., $|\boldsymbol{\psi}_k| = 1$ to satisfy the passive RIS assumption. 

For given $\mathbf{f}$ and $\mathbf{\Phi}$, receive signal-to-noise ratio (SNR) is %can be determined as 
\begin{equation}
    \Gamma(\mathbf{f},\mathbf{\Phi}) = \gamma |\mathbf{h}^T\mathbf{\Phi Hf} + \mu \mathbf{g}^T\mathbf{f}|^2, \label{SNR} % + \mathbf{g}^T\mathbf{f}
\end{equation}
where $\gamma=\frac{P_s (d_1 d_2)^{-\alpha}}{\sigma_n^2}$ and $\mu = (\frac{d_0}{d_1 d_2})^{-\alpha/2}$ is the path loss ratio. The capacity maximization problem to obtain the transmit beamformer $\mathbf{f}$ and RIS phase shift matrix $\mathbf{\Phi}$ is  
\begin{subequations}
\begin{align}
\max_{\mathbf{f},\mathbf{\Phi}} ~ &\log_2( 1 + \Gamma(\mathbf{f},\mathbf{\Phi}) ), \label{objective} \\
{\rm s.t.}~& \mathbf{f} \in \begin{cases}
     \mathcal{B}~;~~~\rm{FD~architecture} \\  \mathcal{L}~;~~~\rm{FA~architecture} 
\end{cases}, \label{constraint_f}\\
& |\boldsymbol{\psi}_k|=1;~\forall k=1,\dots,N,\label{constraint_phi}%
\end{align}\label{optimization_problem}%
\end{subequations} 
where \eqref{constraint_f} and \eqref{constraint_phi} represent the constraints on the transmit beamforming (FA/FD) and RIS phase shift vector, respectively. 
% For the sake of easy referencing, we will refer to the jointly optimal transmit beamforming and RIS phase shift solution as optimal beamforming. 
To evaluate the performance of the beamforming and RIS phase shift solution under both the architectures, we evaluate the {\em outage probability} defined as the probability that the received SNR is below threshold $\beta$ and is given by 
\begin{align}
{\rm P}_{\rm out}=\mathbb{P}[\Gamma(\mathbf{f}_{\rm opt},\boldsymbol{\psi}_{\rm opt})<\beta].\label{outage_prob}
\end{align}
\section{Optimal Beamforming}
\label{sec:optimal_beamforming}
In this section, we present the joint transmit beamformer and RIS phase shift matrix solution that maximizes the capacity \eqref{optimization_problem}. In particular, we focus on obtaining the maximum SNR in closed form under perfect CSI assumption.% for a given channel
\subsection{Optimal RIS phase shifts}
In this subsection, we obtain the optimal RIS phase shift matrix that maximizes the capacity \eqref{objective} for a given transmit beamforming vector $\mathbf{f}$. 
Since logarithm is a monotonically increasing function, the capacity maximization is equivalent to maximizing the  SNR. 
For this, we  rewrite  \eqref{SNR} as
\begin{align}
\Gamma(\mathbf{f},\boldsymbol{\psi}) = \gamma |\boldsymbol{\psi}^T\mathbf{Ef} + \mu \mathbf{g}^T\mathbf{f}|^2,\label{SNR_RIS_SubProb}
\end{align}
where $\mathbf{E}={\rm diag}(\mathbf{h})\mathbf{H}$. 
The received SNR can be maximized by co-phasing the fading coefficients using $\boldsymbol{\psi}$ to maximize the magnitude. Thus, for a given $\mathbf{f}$, the optimal RIS phase shift is
\begin{align}
    \boldsymbol{\psi}^{\rm opt} = \rm{exp}(-\angle{\mathbf{Ef}}+\angle{\mathbf{g}^T\mathbf{f}}). \label{Psi_opt}
\end{align}
Let $\mathbf{G} = \left[\begin{matrix}
    \mathbf{E} ~ \mu \mathbf{g}^T 
\end{matrix}\right]^T$ be a matrix concatenated with DL and IL channel matrices. Substituting \eqref{Psi_opt}, we can rewrite  \eqref{SNR_RIS_SubProb} as  
\begin{align}
\Gamma(\mathbf{f},\boldsymbol{\psi}^{\rm opt}) = ( \|\mathbf{Ef}\|_1 + \mu\|\mathbf{g}\mathbf{f}\|_1 )^2 = \|\mathbf{Gf}\|_1^2. \label{SNR_L1}
\end{align}
\subsection{Optimal Transmit Beamforming}
Using \eqref{SNR_L1}, it can be seen that the capacity maximization problem \eqref{optimization_problem} with optimal $\boldsymbol{\psi}^{\textrm{opt}}$ reduces to the selection of the transmit beamforming vector $\mathbf{f}$ that maximizes the $L_1$ norm of $\mathbf{Gf}$.
We solve this problem for FD and FA beamforming solutions in the following subsections.
\subsubsection{Digital Beamformer}
The $L_1$ norm maximization problem \eqref{SNR_L1} to obtain the optimal transmit beamformer for FD architecture is defined as
\begin{align}
\max_{\|\mathbf{f}\|_2 = 1} \|\mathbf{Gf}\|_1. \label{optimization_f_DB}
\end{align}\textit{In essence, the transmit beamformer problem has now effectively reduced to estimating $L_1$ principal components of the concatenated channel matrix $\mathbf{G} = [({\rm diag}(\mathbf{h})\mathbf{H})^T \mu \mathbf{g}]^T$.}
%{\color{blue}Using the technique from \cite{L1Complex_PCA} that reformulates the original $L_1$ norm maximization problem into two independent maximization problems, we propose an iterative solution for obtaining the optimal beamformer as follows. }
In \cite{L1Complex_PCA}, a novel approach is proposed that breaks such a $L_1$-complex PCA problem into two independent tractable optimization problems that can be solved iteratively. 
The $L_1$ norm of $\mathbf{Gf}$ can be represented from \cite{L1Complex_PCA} as 
\begin{align*}
    \|\mathbf{Gf}\|_1=\max_{\mathbf{u} \in \mathcal{U}^{N+1}} \operatorname{Re}\{\mathbf{u}^H \mathbf{Gf}\},
\end{align*}
where $\mathcal{U}^{N+1} \overset{\Delta}{=} \{\mathbf{u} \in \mathbb{C}^{N+1} : |\mathbf{u}_i| = 1; \forall i\}$ is a unimodular vector space and $\mathbf{u}^{\textrm{opt}}=\exp(-\angle\mathbf{Gf})$. 
Thus,  \eqref{optimization_f_DB} becomes
\begin{align}
 \max_{\|\mathbf{f}\|_2 = 1} \|\mathbf{Gf}\|_1 &= \max_{\|\mathbf{f}\|_2 = 1} \max_{\mathbf{u} \in \mathcal{U}^{N+1}} \operatorname{Re}\{\mathbf{u}^H \mathbf{Gf}\}, \nonumber \\%\label{Reformoptimization1_f_DB}
 &\stackrel{(a)}{=} \max_{\mathbf{u} \in \mathcal{U}^{N+1}} \max_{\|\mathbf{f}\|_2 = 1} \operatorname{Re}\{\mathbf{u}^H \mathbf{Gf}\}, \label{Reformoptimization2_f_DB}
\end{align}
where Step (a) follows from the fact that the maximization operators can be exchanged for the linear objective.
For a fixed $\mathbf{u}$, the optimal transmit beamformer becomes
\begin{align}
    \mathbf{f}^{\rm{opt}} = \frac{\mathbf{G}^H\mathbf{u}}{\|\mathbf{G}^H\mathbf{u}\|}. \label{DB_optimal_f}
    %\mathbf{f}^{\rm{opt}} = {\mathbf{G}^H\mathbf{u}}/{\|\mathbf{G}^H\mathbf{u}\|}. \label{DB_optimal_f}
\end{align}
Thus,  RHS of \eqref{Reformoptimization2_f_DB} becomes $\max_{\mathbf{u} \in \mathcal{U}^{N+1}} \operatorname{Re}\{\mathbf{u}^H \mathbf{G}\mathbf{f}^{\rm{opt}}\}$  which is maximized by $\mathbf{u}$ as 
\begin{align}
    \mathbf{u}^{\rm{opt}} = \exp\left(-\angle\mathbf{G}\mathbf{f}^{\rm{opt}}\right).\label{DB_optimal_u}
\end{align}
Finally, the optimal beamformer $\mathbf{f}^{\rm opt}$ can be obtained by evaluating \eqref{DB_optimal_f} and \eqref{DB_optimal_u} iteratively untill \eqref{optimization_f_DB} converges as summarized in Algorithm \ref{Alg_DB}. Additionally, it is worth noting that the algorithm has a complexity of $\mathcal{O}(NM)$ per iteration, which is due to the evaluation of $\mathbf{f}$ in Step 1. 
% \begin{frame}{}
    \begin{algorithm}\label{Alg_DB}
  % \algsetup{linenosize=\Big}
  % \scriptsize
% \SetKwComment{Comment}{$\triangleright$\ }{}
\KwInput{$\mathbf{g}$, $\mathbf{E}$ and $\mathbf{G}$.}
\KwInit{$\mathbf{u}$ , $\mathbf{f}$}
\SetKwRepeat{Repeat}{Repeat}{Untill: $\operatorname{Re}\{\mathbf{u}^H \mathbf{Gf}\}$ converges}
\Repeat{}{
$\mathbf{f} = \frac{\mathbf{G}^H\mathbf{u}}{\|\mathbf{G}^H\mathbf{u}\|}$,\\
$\mathbf{u} = \rm{exp}(\angle{\mathbf{G}\mathbf{f}})$,}
$\mathbf{f}^{\rm{opt}}=\mathbf{f}$ and $\boldsymbol{\psi}^{\rm opt} = \rm{exp}(-\angle{\mathbf{Ef}^{\rm{opt}}}+\angle{\mathbf{g}^T\mathbf{f}^{\rm{opt}}}).$
\caption{\small{Digital Beamforming Algorithm}}
\end{algorithm}
% \end{frame}
% \begin{remark}
%    By reducing the SNR into the L1-norm form, we have reduced the computational complexity of solving the capacity maximization problem from $\mathcal{O}(N^{9/2}\log(1/\epsilon))$ to $\mathcal{O}(MN)$, where the high complexity arises due to the use of SDR technique with a solution accuracy of $\epsilon$. It is important to note here that the SDR technique has been extensively used throughout the literature.  
% \end{remark}
\begin{remark}
    It is worth noting that the existing optimal beamforming solutions for FD architecture rely on SDR for selecting the RIS phase shift matrix, which has a computational complexity of $\mathcal{O}(N^{9/2}\log(1/\epsilon))$ \cite{SDR}. However, by reducing the SNR into the $L_1$-norm form, we have reduced the computational complexity of the optimal beamforming problem to $\mathcal{O}(MN)$, which is much lower than SDR-based solutions. 
\end{remark}
\subsubsection{Analog Beamformer}
The transmit beamformer under FA architecture is defined as 
\begin{align}
\max_{|\mathbf{f}_m|=\frac{1}{\sqrt{M}}} ~ &|(\boldsymbol{\psi}^T\mathbf{E} + \mu \mathbf{g}^T)\mathbf{f}|^2.\label{optimization_f_AB}
\end{align}
Consequently,  $\mathbf{f}$ that maximizes \eqref{optimization_f_AB} for the given $\boldsymbol{\psi}$  is
\begin{align}
    \mathbf{f}^{\rm{opt}} = e^{-j(\angle{\boldsymbol{\psi}^T\mathbf{E}} - \angle{\mathbf{g}})}/\sqrt{M}.\label{AB_optimal_f}
\end{align}
However, the obtained closed-form expressions for $\mathbf{f}$ and $\boldsymbol{\psi}$, given in \eqref{Psi_opt} and \eqref{AB_optimal_f}, are dependent on each other. Hence, these solutions are iteratively evaluated until \eqref{optimization_f_AB} converges as summarized in Algorithm \ref{Alg_AB}. Algorithm \ref{Alg_AB} has a complexity of $\mathcal{O}(MN)$ per iteration, arising due to matrix multiplication in Steps 1 and 2.
% \begin{frame}{}
\begin{algorithm}\label{Alg_AB}
  % \algsetup{linenosize=\small}
  % \scriptsize
\KwInput{$\mathbf{g}$ and $\mathbf{E}$.}
\KwInit{$\mathbf{u}$ , $\mathbf{f}$}
\SetKwRepeat{Repeat}{Repeat}{Untill: $\Gamma(\mathbf{f},\mathbf{\Phi})$ converges}
\Repeat{}{
$\boldsymbol{\psi} = e^{-j(\angle{\mathbf{Ef}} - \angle{\mathbf{g}})}$,\\

$\mathbf{f} = \frac{1}{\sqrt{M}}e^{-j(\angle{\boldsymbol{\psi}^T\mathbf{E}} - \angle{\mathbf{g}})}$,}
$\mathbf{f}^{\rm{opt}}=\mathbf{f}$ and $\boldsymbol{\psi}^{\rm opt} = \boldsymbol{\psi}.$
\caption{\small{Analog Beamforming Algorithm}}
\end{algorithm}
% \end{frame}
The performance characterization of such a RIS-aided system with FA architecture has not been investigated so far. Motivated by this, we study the capacity and outage performance as follows.
We begin by simplifying the maximum achievable SNR for FA architecture using \eqref{SNR_L1} as
\begin{align}
\max_{|\mathbf{f}_m| =\frac{1}{\sqrt{M}} } \|\mathbf{Gf}\|_1^2. \label{optimization_f_AB_1}
\end{align}
The above objective function  can be upper-bounded as 
\begin{align*}
\|\mathbf{Gf}\|_1 &= \sum\nolimits_{n = 1}^N \left|\sum\nolimits_{m = 1}^M \mathbf{G}_{nm}\mathbf{f}_m\right|,\\
&\stackrel{(a)}{\leq} \sum\nolimits_{n = 1}^N \sum\nolimits_{m = 1}^M \frac{| \mathbf{G}_{nm}|}{\sqrt{M}}, \\
&= \frac{\|\mathbf{G}\|_{1,1}}{\sqrt{M}}, 
\end{align*}
where $\|\cdot\|_{1,1}$ denotes the $L_{1,1}$ norm. 
Here, Step (a) follows from the triangle inequality ($|a+b|\leq |a|+|b|$) and the analog beamforming constraint $|\mathbf{f}_m|=\frac{1}{\sqrt{M}}$. We summarize this upper bound in the following theorem. 
\begin{theorem}
\label{thm_SNRmax}
The maximum capacity of a RIS-aided MISO communication system with FA architecture is upper bounded by $\log_2(1+\Gamma^{\rm UB})=\log_2(1+\frac{\gamma}{M}\|\mathbf{G}\|^2_{1,1})$.  %$\frac{1}{M}\|{\rm diag}(\mathbf{h})\mathbf{H}\|^2_1$. 
\end{theorem}
\begin{corollary}\label{cor_SNRmax}
  The maximum capacity upper bound of the RIS-aided MISO downlink given in Theorem \ref{thm_SNRmax} under both FA and FD architectures reduces to $\log_2(1+\frac{\gamma}{M}\|\mathbf{\bar{G}}\|^2_1)=\log_2(1+( N+\mu)^2M\gamma)$ for an LoS-channel.
\end{corollary}
\begin{IEEEproof}
    In the presence of an LoS channel, we have $\mathbf{g} = \kappa_l \mathbf{\bar{g}}$, $\mathbf{h} = \kappa_l \mathbf{\bar{h}}$, and $\mathbf{H} = \kappa_l \mathbf{\bar{H}}$, corresponding to  $K_o=K_1=K_2=\infty$ in Rician channel model discussed in Section \ref{system_model}.
    Thus, matrix $\mathbf{G}$ becomes equal to $\mathbf{\bar{G}}=[(\rm{diag}(\mathbf{\bar{h}})\mathbf{\bar{H}})^T~ \mu \mathbf{\bar{g}}]^T$.
    For such a case, the channel envelope with optimal RIS phase shift $\boldsymbol{\psi}_{\rm opt}$ \eqref{SNR_L1}, and a given FA beamformer $\mathbf{f}$ is simplified as 
    \begin{align}
    \|\mathbf{\bar{G}f}\|_1&=\sum\nolimits_{n=1}^{N+1}\big|\sum\nolimits_{m=1}^M \mathbf{\bar{G}}_{nm}\mathbf{f}_m\big|,\nonumber\\
        &\stackrel{(a)}{=}\sum\nolimits_{n=1}^{N+1}\big|\sum\nolimits_{m=1}^M \mathbf{\bar{G}}_{n1}c_m\mathbf{f}_m\big|,\nonumber\\
        &\stackrel{(b)}{=}\frac{1}{\sqrt{M}}\sum\nolimits_{n=1}^{N+1}\sum\nolimits_{m=1}^M| \mathbf{\bar{G}}_{n1}|,\nonumber\\
        &\stackrel{(c)}{=}( N + \mu)\sqrt{M}, \label{LoS_SNRmax}
    \end{align}   
    where Step (a) follows from the fact that the columns of matrix $\mathbf{\bar{G}}$ are linearly dependent, Step (b) follows from the unit modulus constraint on the analog beamformer $\mathbf{f}$ such that it maximizes the sum by setting $\mathbf{f}_m=\frac{1}{\sqrt{M}}c_m^H$ and the fact that $|c_m|=1$, and Step (c) follows from $|\mathbf{\bar{G}}_{n1}|=1$. 

    Moreover, the equality presented in \eqref{LoS_SNRmax} also holds for digital beamforming. This is because the inner summation given in Step (b)  is also maximized for the digital beamformer $\mathbf{f}_m=\frac{1}{\sqrt{M}}c_m^H$ such that $\|\mathbf{f}\|^2=1$.
\end{IEEEproof}
\begin{corollary}\label{cor_SNRmax2}
 The upper bound given in Corollary \ref{cor_SNRmax} is achievable by both FD and FA architectures in the absence of DL under dominant LoS and is given by $\log_2(1+\gamma N^2M)$.     
\end{corollary}
\begin{proof}
Let $\mathbf{\Bar{E}} = \textrm{diag}(\Bar{\mathbf{h}})\Bar{\mathbf{H}}$. Thus, the SNR given in \eqref{SNR_L1} in the absence of DL under dominant LoS channel becomes $$\gamma\|\mathbf{\Bar{E}f}\|_1^2  = \gamma \|\textrm{diag}(\mathbf{a}_N(\theta_{\rm{rd}}))\mathbf{a}_N(\theta_{\rm ra})\|_1^2 |\mathbf{a}_M(\theta_{\rm bd}^i)^T\mathbf{f}|^2.$$
The optimal choice of $\mathbf{f}$ under the FA and FD architectures is $\mathbf{f}^{\textrm{opt}}_{\textrm{AB}} = e^{-j\angle{\mathbf{a}_M(\theta_{\rm bd}^i)}} / \sqrt{M}$ and $\mathbf{f}^{\textrm{opt}}_{\textrm{DB}} = \mathbf{a}^H_M(\theta_{\rm bd}^i)/\|\mathbf{a}^H_M(\theta_{\rm bd}^i)\|$, respectively. For these optimal choices, the SNR becomes $$\gamma\|\mathbf{\bar{E}f}^{\textrm{opt}}\|_1^2= \gamma N^2M,~~~\text{for}~\mathbf{f}^{\textrm{opt}}\in\{\mathbf{f}^{\textrm{opt}}_{\textrm{AB}},\mathbf{f}^{\textrm{opt}}_{\textrm{DB}}\}.$$ Thus, it can be seen that the upper bound proposed in Corollary \ref{cor_SNRmax} is achievable by both the architectures in a LoS dominated scenario without DL.
\end{proof}
\begin{remark}\label{remark2}
    From Corollary \ref{cor_SNRmax2}, it can be safely deduced that the capacity upper bound presented in Theorem \ref{thm_SNRmax} becomes tight in the absence of DL (i.e., $\mathbf{G} = \mathbf{E}$) under scenarios including strong LoS components (i.e., for a large $K_i$). In other words, the capacity upper bound in the absence of DL becomes tighter as the rank of $\mathbf{E}$ becomes low, which further reduces to equality when $\mathbf{E}$ becomes unit rank, i.e., $K_i\to\infty$. Moreover, the capacity upper bound in the presence of DL can also be achieved when the angle between DL and IL is very small and $K_i$ is large (in which case $\mathbf{G}$ becomes unit rank).   
\end{remark}
\section{Outage Probability Analysis}
\label{sec:outage_analysis}
The outage performance characterization of RIS-aided FD/FA systems in the absence of DL is as follows.
The outage probability defined in \eqref{outage_prob} is lower bounded as 
\begin{align}   \rm{P_{out}} \geq \rm{P_{out}^{LB}} = \mathbb{P}\left[\Gamma^{\rm UB}\leq \beta\right] .\label{outage_prob_LB}
\end{align}
    Using Theorem \ref{thm_SNRmax}, the SNR upper bound in the absence of DL becomes $\Gamma^{\textrm{UB}}=\gamma\|\mathbf{E}\|_{1,1}^2$. Thus, we can write 
    \begin{align}   \rm{P_{out}^{LB}}=\mathbb{P}\left[\|\mathbf{E}\|_{1,1}\leq \sqrt{\beta/\gamma}\right].\label{outage_prob_LB_new}
\end{align}
Recall that the above lower bound becomes tighter for a low-rank channel matrix $\mathbf{E}$, which further reduces to equality for $\mathbf{E}$ with rank 1 as highlighted in Remark \ref{remark2}.
Note that $\|\mathbf{E}\|_{1,1}$ is the sum-of-product of Rician random variables, making it difficult to derive outage probability directly. Thus, we first obtain its {\rm MGF} and then use it to determine $\rm{P_{out}^{LB}}$.
\begin{theorem}\label{Theorem_OP}
    The MGF of $\|\mathbf{E}\|_{1,1}$ is 
    \begin{align}
    M(-s)=s^{MN} \left[\int_0^\infty g(h)^M f_{|\mathbf{h}_n|}(h){\rm d}h\right]^N,\label{MGF_gamma_UB}
    \end{align}
    where $g(h)=\mathcal{L}\left(1-Q_1\left(\frac{\kappa_l}{\kappa_n},\frac{x}{|\mathbf{h}_n|\kappa_n}\right)\right)$, $Q_1(\cdot)$ is the Marcum-Q function, $\mathcal{L}(\cdot)$ is Laplace transform (LT), and $f(\cdot)$ is Rician density function.
\end{theorem}
\begin{IEEEproof} 
Let us define
    \begin{align}
    Y=\|\mathbf{E}\|_{1,1}=\sum\nolimits_{n=1}^N Y_n,\label{G_L1_norm}
\end{align}
where $Y_n=|\mathbf{h}_n|\sum_{m=1}^M|\mathbf{H}_{nm}|$.
     
We begin by writing the MGF of the channel gain $|\mathbf{H}_{nm}|$ by using the differentiation property of LT as  %{\color{red}\cite[Eq. X]{MGF_ref}}
    \begin{align}
        M_{|\mathbf{H}_{nm}|}(-s)&=s\mathcal{L}(F_{|\mathbf{H}_{nm}|}(x)),\nonumber %\label{MGF_Hnm}
    \end{align}
    where $\mathcal{L}(\cdot)$ represents LT and $F_{|\mathbf{H}_{nm}|}(x)=1-Q_1(\frac{\kappa_l}{\kappa_n},\frac{x}{\kappa_n})$ is the cumulative distribution function ({\rm CDF}) of $|\mathbf{H}_{nm}|$.
    Let $X_n=\sum_{m=1}^M|\mathbf{H}_{nm}|$. The $n$-th term of $Y$ given in \eqref{G_L1_norm} becomes $Y_n=|\mathbf{h}_n|X_n$ and thus its CDF is determined as
    \begin{align}
        F_{Y_n}(y)&=\int_0^\infty \int_0^{\frac{y}{|\mathbf{h}_n|}}f_{|\mathbf{h}_n|}(h)f_{X_n}(x){\rm{d}}x{\rm{d}}h,\nonumber\\
        &=\int_0^\infty f_{|\mathbf{h}_n|}(h)F_{X_n}(y/h){\rm{d}}h. \label{CDF_Yn}
    \end{align}
    Thus, the MGF of $Y_n$ can be obtained as    
    \begin{align}
        M_{Y_n}(-s)&~~=s\mathcal{L}\left(F_{Y_n}(y)\right),\nonumber\\
        &~~\stackrel{(a)}{=}\int_0^\infty f_{|\mathbf{h}_n|}(h)  s\mathcal{L}(F_{X_n}(y/h)) {\rm{d}}h,\nonumber\\
        &~~\stackrel{(b)}{=}\int_0^\infty f_{|\mathbf{h}_n|}(h)  \left[s\mathcal{L}\left(1-Q_1(\kappa_l/\kappa_n,y/(h\kappa_n))\right)\right]^{M} {\rm{d}}h. \nonumber %\label{MGF_Yn}
    \end{align}
    where Step (a) follows from using \eqref{CDF_Yn} and the fact that LT is a linear operator, Step (b) follows from $M_{X_n}(-s)=s\mathcal{L}(F_{X_n}(x))$ and 
     using the fact that $\mathbf{H}_{nm}$ are independent and identically distributed (i.i.d.) random variables.
    Since $Y_n$s are i.i.d. random variables, we can obtain the MGF of $Y$ as $ M_Y(s)=M_{Y_n}(s)^N$.
    This completes the proof.
\end{IEEEproof}    
Finally, we evaluate the lower bound on the outage probability $\rm{P_{out}^{LB}}$ by first numerically inverting the MGF of $\|\mathbf{E}\|_{1,1}$ given in Theorem \ref{Theorem_OP} and then using it to evaluate \eqref{outage_prob_LB_new}. For numerical inversion of MGF, please refer to  \cite{NumericalInversion_Alouini}. 
\section{Numerical Results and Discussion}
\setcounter{figure}{0}
\begin{figure*}[t!]
\centering
\resizebox{1\columnwidth}{!}{
\includegraphics[width=.33\columnwidth]{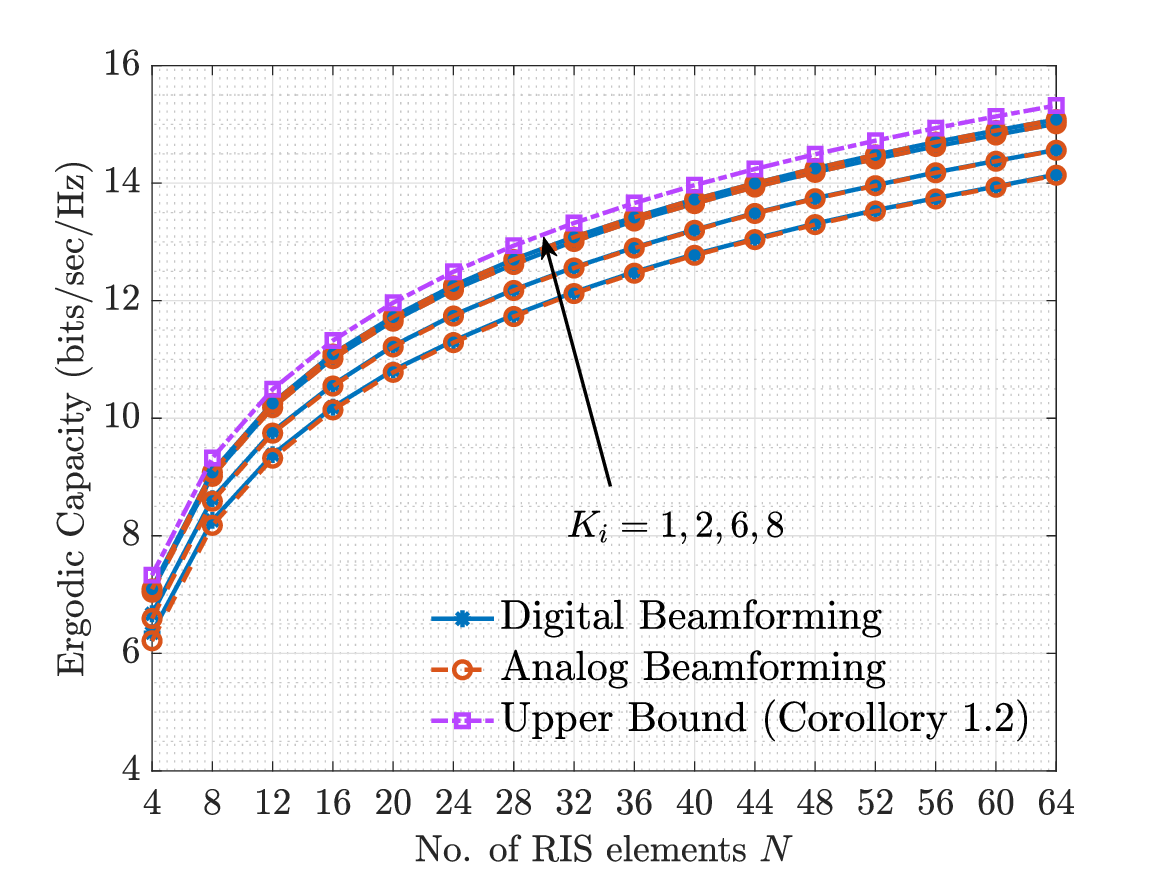} 
\includegraphics[width=.33\columnwidth]{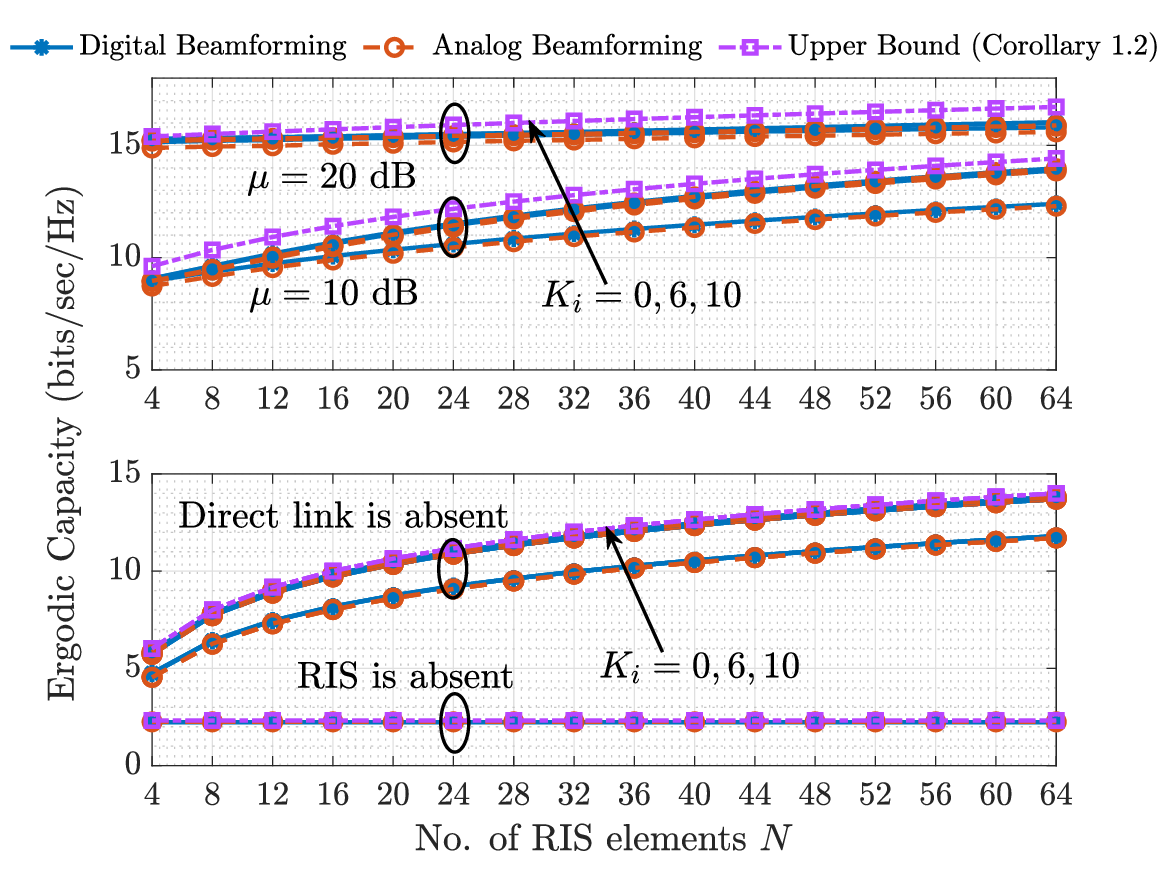}
\includegraphics[width=.33\columnwidth]{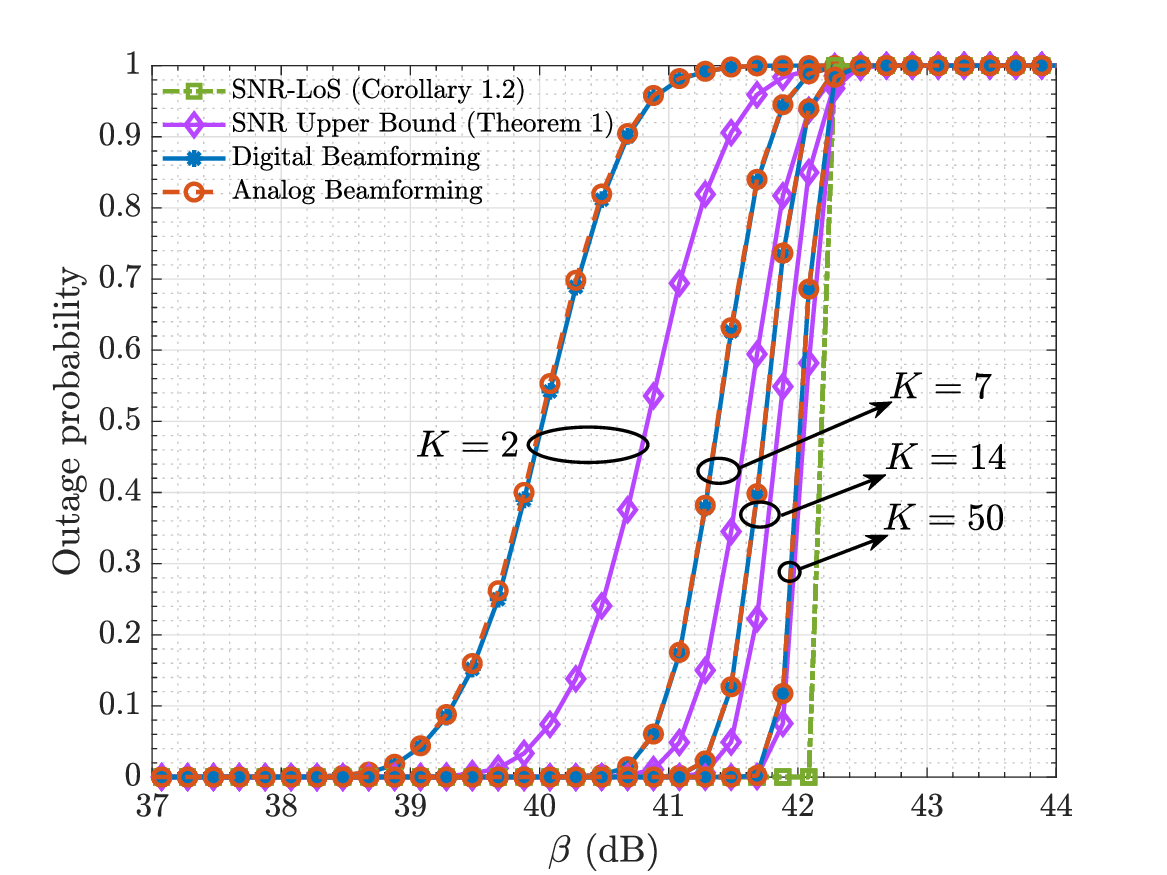}
}
%\caption{Left: Ergodic capacity vs. $N$ with and without direct links, Center: Ergodic capacity vs. $N$ with and without direct-indirect links, Right: Outage probability vs. SNR threshold $\beta~\textrm{(dB)}$.} 
\caption{Capacity vs. $N$  without DL (left),  with either DL or IL (middle bottom), and with both DL and IL (middle top). Outage probability  (right).} 
\label{fig:all}
\end{figure*}
This section presents the numerical analysis of capacity and outage obtained by the proposed algorithms for FD/FA architectures and compares them with the upper bound derived in Theorem \ref{thm_SNRmax}. For numerical analysis, we assume the number of BS antennas $M = 4$, the number of RIS elements $N = 64$, path loss ratio $\mu = 5~\textrm{dB}$, and $\gamma = 1$.
Figure \ref{fig:all} (Left) shows the capacity as a function of $N$ under both FD and FA architectures. It can be observed that the capacity increases with $N$ as well as with the Rician factor $K_i$. This is expected as larger $N$ provides larger array gains and larger $K_i$ provides stronger LoS paths, ensuring better capacities. Further, it can be seen that the capacity is close to the derived upper bound for reasonable values of $K_i$ which further becomes tight with increasing $K_i$. This is because the channel becomes dominated with the LoS component for very large $K_i$, and for the LoS channel, the upper bound reduces to equality as mentioned in Corollary \ref{cor_SNRmax}.
Besides, we can also observe that the AB and DB  perform almost equally as $K_i$ becomes larger.

Figure \ref{fig:all} (Middle) shows the comparison of achievable capacity under three schemes, namely 1) the absence of DL at $\mu=-\infty$, 2) the presence of DL at $\mu=10$ and $20$ dB, and 3) absence of RIS, i.e., IL. We first observe that the derived upper bound is tight in all the schemes, especially when $K_i$ is large. Next, we see that the capacity improves with $\mu$, which strengthens DL, which in turn provides additional spatial diversity to achieve a higher capacity. However, it is noteworthy that the capacity improves slowly with increase in $N$ when $\mu$ is large. This is mainly because, in the presence of strong DL, the improvement in receive SNR due to RIS is not significant. Moreover, as capacity is a logarithmic function of SNR, we observe a saturation in capacity with $N$. 
% {\color{blue} Note on gap increasing w.r.t. $N$ when $\mu$ is large is needed.}
Another interesting observation can be made at large $\mu$ and increasing $N$, which shows the proposed capacity upper bound becoming loose. This is because, the concatenated matrix $\mathbf{G} = [\mathbf{E}^T \mathbf{g}^T]^T$ becomes rank 2 at large values of $\mu$, even in a dominant LoS propagation scenario. This verifies Corollary \ref{cor_SNRmax2} further that the upper bound is exact in the absence of DL under a dominant LoS scenario.   
%\begin{bmatrix}\mathbf{E}\\ \mathbf{g}\end{bmatrix}

Figure \ref{fig:all} (Right) shows the outage performance of the proposed beamforming schemes under FA and FD architectures in the absence of DL at various values of $K_i$. We first observe that the outage performance of both architectures is very close. 
% These performances become closer to the derived outage lower bound (presented in Section \ref{sec:outage_analysis}) as $K_i$ increases.
As $K_i$ increases, these performances increasingly approach the derived outage lower bound presented in Section \ref{sec:outage_analysis}.
It can also be observed that the distributions of exact SNR under both the architectures and the SNR upper bound converge to the deterministic value derived in Corollary \ref{cor_SNRmax2}, i.e. $\Gamma(\mathbf{f}^\textrm{opt},\mathbf{\psi}^\textrm{opt})=\Gamma^{\rm UB}=\gamma MN^2$ as $K_i \to \infty$.  
\section{Conclusion}
This letter investigated optimal beamforming for maximizing the capacity of RIS-aided downlink systems with FD and FA architectures in the presence of Rician faded DL and IL. We first showed that the capacity maximization problem reduces to an $L_1$-norm maximization problem with respect to the transmit beamformer after optimally configuring the RIS. We proposed a complex $L_1$-PCA-based algorithm to obtain the optimal FD beamformer.
% {\color{red}This problem is solved by reformulating the $L_1$ norm maximization objective as maximizing the projection of the beamformer onto unimodular vector space as proposed in \cite{L1Complex_PCA}.} 
%We propose a low-complexity algorithm to obtain the FD beamformer. 
We proposed another algorithm to obtain the optimal FA beamformer with low complexity. Both the proposed algorithms iterate over two closed-form expressions. To characterize the performance of the proposed algorithms, we derived an upper bound on the capacity and analyzed its corresponding outage performance. Specifically, we derived the MGF of the envelope of SNR upper bound, which we then numerically inverted to obtain the outage probability lower bound. Moreover, we analytically established that the proposed bounds on capacity and outage become exact when the channel matrix becomes unit rank, i.e., LoS components of DL and IL are strong and aligned. 
\bibliographystyle{IEEEtran}
\bibliography{Ref}
\end{document}